\documentclass[letterpaper, 10 pt, conference]{ieeeconf}  

\IEEEoverridecommandlockouts                              

\usepackage{xcolor}
\colorlet{citeblue}{blue!50!black}
\colorlet{linkred}{red!50!black}
\usepackage[isbn=false,
url=true,
doi=false,
backend=biber,
style=ieee,
bibstyle=ieee, 
citestyle=numeric-comp, 
sortcites=true, 
maxbibnames=10,
giveninits=true, 
backref=false]{biblatex}
\usepackage{booktabs}
\usepackage{graphics} 
\usepackage{epsfig} 
\usepackage{times} 
\usepackage{amsmath} 
\usepackage{amssymb}  

\usepackage{amsthm}
\usepackage{amsfonts}
\usepackage{mathtools}
\usepackage{amsthm}
\usepackage{nicefrac}
\usepackage{marvosym}
\usepackage{tikz}
\usepackage[hidelinks]{hyperref}

\newtheorem{lemma}{Lemma}

\newtheorem{proposition}[lemma]{Proposition}

\newtheorem{assumption}[lemma]{Assumption}

\usepackage{algorithm}
\usepackage{algpseudocode}

\usepackage{siunitx}
\newcommand*{\ctrl}{\boldsymbol{\theta}}

\newcommand*{\ctrldomain}{\Theta}

\newcommand{\argmax}{\arg\max}
\newcommand{\expm}[1]{#1}

\algrenewcommand\algorithmicindent{0.4em}

\newcommand\copyrighttext{%
  \footnotesize \textcopyright 2024 IEEE. Personal use of this material is permitted.
  Permission from IEEE must be obtained for all other uses, in any current or future
  media, including reprinting/republishing this material for advertising or promotional
  purposes, creating new collective works, for resale or redistribution to servers or
  lists, or reuse of any copyrighted component of this work in other works.}
\newcommand\copyrightnotice{%
\begin{tikzpicture}[remember picture,overlay]
\node[anchor=south,yshift=10pt] at (current page.south) {\parbox{\dimexpr\textwidth-\fboxsep-\fboxrule\relax}{\copyrighttext}};
\end{tikzpicture}%
}

\title{\LARGE \bf
Lipschitz Safe Bayesian Optimization for Automotive Control
}

\author{Johanna Menn$^{1}$, Pietro Pelizzari$^{2}$, Michael Fleps-Dezasse$^{2}$, Sebastian Trimpe$^{1}$
\thanks{$^{1}$J.~Menn and S.~Trimpe are with the Institute for Data Science in Mechanical Engineering (DSME), RWTH Aachen University, 52068 Aachen, Germany. E-mail:
\{johanna.menn, trimpe\}@dsme.rwth-aachen.de}%
\thanks{$^{2}$P.~Pelizzari and M.~Fleps-Dezasse are with ZF Engineering Solutions, ZF Friedrichshafen AG, 88045 Friedrichshafen, Germany. E-mail: \{pietro.pelizzari, michael.fleps-dezasse\}@zf.com}%
\thanks{This work was performed in part within the Helmholtz School for Data Science in Life, Earth and Energy (HDS-LEE), and in part funded by the German Federal Ministry for Economic Affairs and Climate Action (BMWK) through the project EEMotion and by the Deutsche Forschungsgemeinschaft (DFG, German Research Foundation) under Germany's Excellence Strategy -- EXC-2023 Internet of Production -- 390621612. Computations were performed with computing resources granted by RWTH Aachen University under project rwth1576.}
}

\begin{document}

\maketitle
\copyrightnotice
\thispagestyle{empty}
\pagestyle{empty}

\begin{abstract}

Controller tuning is a labor-intensive process that requires human intervention and expert knowledge. Bayesian optimization has been applied successfully in different fields to automate this process. However, when tuning on hardware, such as in automotive applications, strict safety requirements often arise.  
To obtain safety guarantees, many existing safe Bayesian optimization methods rely on assumptions that are hard to verify in practice. This leads to the use of unjustified heuristics in many applications, which invalidates the theoretical safety guarantees. 
Furthermore, applications often require multiple safety constraints to be satisfied simultaneously. Building on recently proposed Lipschitz-only safe Bayesian optimization, we develop an algorithm that relies on 
readily interpretable assumptions and satisfies multiple safety constraints at the same time.
We apply this algorithm to the problem of automatically tuning a trajectory-tracking controller of a self-driving car. 
Results both from simulations and an actual test vehicle underline the algorithm's ability to learn tracking controllers without leaving the track or violating any other safety constraints.

\end{abstract}

\section{INTRODUCTION} 
Tuning controllers in the automotive domain is a time-consuming task that requires expert knowledge. Different cars have individual dynamic behaviors due to varying configurations of drive trains, wheels, and damping components. Therefore, fine-tuning is necessary when deploying controllers for different vehicles. Automating this task can markedly reduce time and cost of application engineering.

A popular method for automatic controller tuning is Bayesian optimization (BO) \cite{Alonso2020, stenger2022benchmark, paulson2023tutorial}.
This black-box optimization approach is sample efficient, which makes it suitable for controller tuning, where experiments are expensive. The goal of BO is to find an optimum of an unknown objective function given only noisy observations of that function. The objective function is sequentially queried based on a sample selection criterion to find the optimum. In controller tuning, the objective function can be a performance measure, which generally depends on the controller parameters and can be flexibly specified based on the control loop requirements. In each optimization step, an experiment is conducted to determine the performance of a specific parameter set. 

When tuning on hardware, it is often vital to respect multiple safety requirements.
These safety constraints must be considered in the tuning process and must not be violated in any experiment. 
Safe BO accounts for such constraints in the optimization and only queries parameter sets that are believed to be safe.

Conceptually, any learning method with safety guarantees has to rely on assumptions, including safe BO.

While assumptions by their very nature typically can not be \emph{proven}, practitioners have to be able to \emph{judge} whether an assumption is realistic or not for a specific application.  In particular, this requires assumptions to be 
interpretable by domain experts. 
However, popular safe BO algorithms require theoretical assumptions on the objective function, which are not readily interpretable.
In existing applications of safe BO, this often leads to using heuristics, invalidating the theoretical safety guarantees. 
Furthermore, deterministic guarantees are often preferred in safety-critical applications, while common safe BO provides probabilistic safety.

\begin{figure}
    \centering
    \input{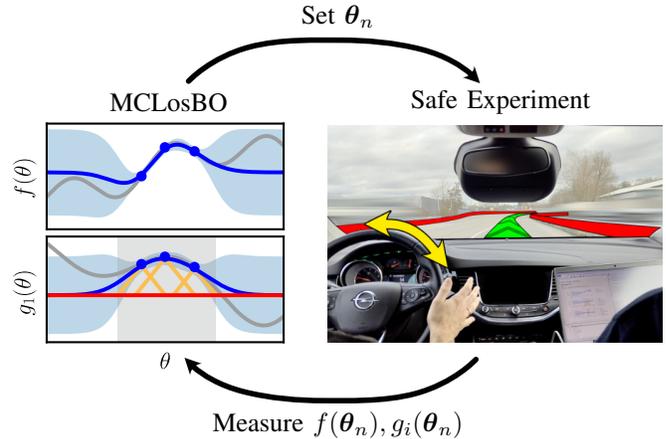}
    \caption{MCLoSBO is used to tune the parameters $\ctrl$ of a lateral controller that is steering the vehicle (yellow arrow) to track the target trajectory (green) as closely as possible. The tracking performance is measured by $f(\ctrl)$. The optimization is sequential: MCLoSBO sets new safe parameters $\ctrl_n$ that do not violate the safety constraints described by $g_i(\ctrl)$ (red). After one lap on the test track, the corresponding values of the performance function $f(\ctrl_n)$ and safety functions $g_i(\ctrl_n)$ are measured with noise.}
    \label{fig:summary}%
    \vspace{-0.5cm}
\end{figure}%

In this work, we present \emph{Multiple Constraints Lipschitz-only Safe Bayesian Optimization (MCLoSBO)}, a safe BO method that considers multiple safety constraints and is based on deterministic and interpretable assumptions. In particular, the safety guarantees rely on common Lipschitz assumptions on the unknown constraints and objective functions, as well as worst-case noise bounds. Our method can be applied to different controller tuning tasks with safety requirements. 

For the demonstration of MCLoSBO, we focus on tuning a lateral trajectory tracking controller that steers an autonomous vehicle (see Fig.~\ref{fig:summary}). We apply the developed method in simulation to study the effects of the practical adaptations, i.e. hyperparameter optimization and the use of batch optimization, and compare it against the state-of-the-art algorithm SafeOpt-MC. Finally, we apply the method to tune the parameters of the controller on a test vehicle. We show the initial controller can be safely improved in simulation and on the vehicle. In summary, the contributions of this work are:
\begin{itemize}
\item the MCLoSBO algorithm, which provides guarantees on satisfaction of multiple safety constraints during tuning while relying on interpretable assumptions (Lipschitz constants and noise bounds);
\item a simulation benchmark which shows that MCLoSBO can safely tune parameters and perform on par with the state-of-the-art; and
\item the first real-world implementation of a LoSBO-type algorithm, in which we safely tune a lateral controller in an autonomous car.
\end{itemize}

\section{Related Work}
\label{sec:related_work}
Bayesian Optimization (BO) \cite{garnett} has been used to automate controller tuning in many domains, see \cite{Alonso2020, stenger2022benchmark, paulson2023tutorial}  
and references therein. 
It has also been applied successfully to automotive applications to tune parameters of trajectory tracking and path following controllers \cite{fröhlich, Rupenyan2021, Wischnewski2019}.  
While \cite{fröhlich, Rupenyan2021} apply different types of BO (contextual and constrained) for problems that differ from the one discussed in this paper, \cite{Wischnewski2019} also uses safe BO. 
In contrast to this work, only one parameter is tuned in \cite{Wischnewski2019} on a race car, and some safety violations are reported.
In this work, we develop a new method for safe BO and demonstrate it by tuning a trajectory tracking controller on a real test vehicle. 

Safe BO algorithms aim to find the optimum of the target functions without violating given safety constraints in any optimization step. In contrast, in constrained BO, constraints can be violated during the optimization and must only be fulfilled for the final optimization result \cite{Kim2021}. 
The present safe BO setting was introduced in \cite{sui15} and addresses safe learning with respect to a single safety constraint on the target function. 
Many extensions build upon the original SafeOpt, including GOOSe \cite{Turchetta2019}, SafeLineBO \cite{kirschner2019adaptive}, ISE \cite{Bottero2022}, GoSafeOpt \cite{Sukhija2023}. 
SafeOpt-MC \cite{Berkenkamp2020} and StageOpt \cite{Sui2018} introduce extensions of SafeOpt to multiple safety constraints. 

Safety guarantees in SafeOpt \cite{sui15} and its variants rely on correct frequentist uncertainty bounds \cite{fiedler2024safety}, which must cover the graph of the target function with high probability. Moreover, a known Lipschitz constant of the unknown functions and sub-Gaussian observation noise are typically assumed (with the exception of \cite{Berkenkamp16}, which only relies on uncertainty bounds and sub-Gaussian noise).
Determining correct uncertainty bounds is challenging in practice \cite{fiedler}.
While several uncertainty bounds are available, e.g., \cite{srinivas, chowdhury, fiedler}, they are all based on a bound on the norm of the unknown function in a reproducing kernel Hilbert space (RKHS). 
However, this assumption is hard to interpret or even validate for unknown target functions \cite{fiedler, fiedler2024safety}.   
Therefore, to the best of our knowledge, all applications of safe BO skip the determination of the RKHS norm bound and use heuristics instead, e.g. \cite{Berkenkamp16, baumann2021, König2021, doersche2021}.
While this procedure yields cautious exploration behavior and may work well in practice, such heuristics invalidate the theoretical safety \emph{guarantees} \cite{fiedler2024safety}.
In contrast, this work builds on a recently proposed algorithm LoSBO \cite{fiedler2024safety}, whose safety guarantees are based on a known Lipschitz constant only, but circumvents the need for a computable RKHS norm bound. The Lipschitz constant, in combination with bounded noise, leads to \emph{deterministic} safety guarantees.  
To obtain such guarantees for \emph{multiple} constraints, this paper extends LoSBO, leveraging ideas from SafeOpt-MC \cite{Berkenkamp2020}.
Finally, this work is the first to apply Lipschitz-only Safe BO on a safety-critical real-world problem.

\section{Preliminaries}
We now introduce the necessary preliminaries, provide further details on safe BO, and state the objective of this paper. 
\subsection{Gaussian processes}
In BO, Gaussian processes (GPs) are often used as surrogates for the objective function $f$. 
 GPs are stochastic processes that can be used for nonparametric regression. A GP $h$ is uniquely defined by its mean function $m(\ctrl)=\mathbb{E}[h(\ctrl)]$ and the covariance function $k(\ctrl, \ctrl')=\mathbb{E}[(h(\ctrl)-m(\ctrl))(h(\ctrl')-m(\ctrl'))]$. 
To approximate an unknown function, we start with a GP prior that can be updated with data $\mathcal{D}_n=\{(\ctrl_0, y_0), ...,(\ctrl_n, y_n)\}$. 
This data is obtained by observing the unknown function with identically and independently distributed Gaussian noise $\epsilon \sim \mathcal{N}(0, \sigma^2)$ such that $y_i = f(\ctrl_i)+\epsilon_i$, $i=0,\ldots,n$.  
We denote the vector of noisy observations as $\mathbf{y}=[y_0, ..., y_n]^T$. 
Without loss of generality, we can assume that the mean function is zero, as data can always be transformed accordingly. 
The posterior mean, posterior covariance, and posterior variance are then 
$\mu_n(\ctrl)=\mathbf{k}_n(\ctrl)^T(\mathbf{K}_n + \sigma^2\mathbf{I}_n)^{-1} \mathbf{y}_n$, 
$k_n(\ctrl, \ctrl')=k(\ctrl, \ctrl')-\mathbf{k}_n(\ctrl)^T(\mathbf{K}_n+\sigma^2\mathbf{I}_n)^{-1}\mathbf{k}_n(\ctrl')$ 
and $\sigma_n^2(\ctrl)=k_n(\ctrl, \ctrl)$. 
The entries of the kernel matrix $\mathbf{K}_n \in \mathbb{R}^{n \times n}$ are $[k(\ctrl, \ctrl')]_{\ctrl, \ctrl' \in \mathcal{D}_n}$, the vector $\mathbf{k}_n(\ctrl)=[k(x_0,x) ... k(x_n, x)]$ describes the covariances between $\ctrl$ and the observed data points, and $\mathbf{I}_n$ is the $n \times n$ identity matrix. We refer to \cite{rasmussen} for a more detailed description.

\subsection{Bayesian optimization}
BO \cite{garnett} is a sequential global black-box optimization method.
The goal is to find the maximum of an unknown function $f:\ctrldomain \rightarrow \mathbb{R}$, i.e.
\begin{equation} 
\label{eq:optproblem}
\ctrl^* = \argmax_{\ctrl \in \ctrldomain}{f(\ctrl)}, 
\end{equation}
where $\ctrldomain \subseteq \mathbb{R}^d$ is the input domain of the unknown function and the decision variable is $\ctrl \in \ctrldomain$. Usually, BO is applied to problems where the evaluation of the target function is expensive, so the number of possible evaluations is limited to a fixed budget of  $N \in \mathbb{N}^+$ iterations. 
A BO algorithm consists of two main components, a surrogate model, often a GP, and the acquisition function $\alpha$.  In each optimization step $n$, the next query $\ctrl_n$ is chosen by maximizing the acquisition function. The acquisition function uses the information from the surrogate model and is usually comparatively cheap to optimize.
The target function $f$ is measured with independent Gaussian noise $\epsilon_{0,n}$,
\begin{equation}
    y_{0,n} = f(\ctrl_n)+\epsilon_{0,n}
    \label{eq:measurement_f}
\end{equation} 
and the GP surrogate model is then updated with the data point $(\ctrl_n, y_{0,n})$.

\subsection{Bayesian optimization with safety constraints}
When applying BO on physical systems, for example, when tuning a controller on a vehicle, certain parameters $\ctrl$ should not be queried as this can lead to safety violations. Safe BO algorithms aim to find the optimum of the target function $f$ without violating these safety constraints in any iteration.
Following \cite{Berkenkamp2020, Sui2018}, we formalize safety by $q \in \mathbb{N}^+$ constraints $g_i(\ctrl) \geq 0$, where $g_i:\ctrldomain \rightarrow \mathbb{R}$ for $i \in \mathcal{I}_g:=\{1,\ldots,q\}$.
The safety threshold is zero without loss of generality as the safety functions can be shifted accordingly.
These safety constraints have to hold at every iteration; that is, we require
\begin{equation}
    \label{eq:safety}
     \ g_i(\ctrl_n) \geq 0 \quad \forall i \in \mathcal{I}_g, n \in \{0, \ldots , N\}.
\end{equation}

For $g_1(\ctrl) = f(\ctrl)$ and $q=1$, this reduces to a single safety constraint on the objective function as in the seminal work \cite{sui15}. 
Just as for objective $f$, the safety functions $g_i$ are unknown and can only be queried through noisy observations
\begin{equation}
\label{eq:measurement_g}
    y_{i,n}=g_i(\ctrl_n)+\epsilon_{i,n} \quad \forall i \in \mathcal{I}_g
\end{equation}
where 
$\epsilon_{i,n}$ are noise variables.
Each constraint function $g_i$ is modeled as a separate GP, which is updated with the data $(\ctrl_n, y_{i,n})$.

The generic multiple constraints safe BO problem (approximately) solves the black-box optimization \eqref{eq:optproblem} through sequential queries $\ctrl_n$ evaluating \eqref{eq:measurement_f} and \eqref{eq:measurement_g}, while with high-probability not violating the safety constraints $g_i$ \eqref{eq:safety} during the whole optimization process. When used for controller tuning, $f$ represents the control objective, while $g_i$ represent safety constraints that must hold throughout the tuning. This will be detailed for the automotive controller tuning problem in Sec.~\ref{sec:application}.

To guarantee safety, safe BO algorithms rely on regularity assumptions on the performance and safety functions and the noise in \eqref{eq:measurement_g}. 
As discussed in Sec.~\ref{sec:related_work}, existing safety guarantees require assumptions that are hard to validate in practice (in particular, a known RKHS norm bound). We want to avoid such assumptions in this work.

\paragraph*{Problem statement}
We seek an algorithm to solve the multi-constraints safe BO problem given by \eqref{eq:optproblem}--\eqref{eq:measurement_g} with deterministic safety guarantee and relying on readily-interpretable geometric assumptions (no RKHS norm bound).

While existing multi-constraints safe BO algorithms SafeOpt-MC \cite{Berkenkamp2020} and StageOpt \cite{Sui2018} address the multi-constraint safety problem \eqref{eq:safety}, they rely on the knowledge of RKHS norm bounds for all $g_i$ (in addition to a Lipschitz assumption, which we also utilize) and yield probabilistic safety guarantees. 

\section{Multiple constraint Lipschitz-only Safe BO}
In this section, we propose the algorithm \emph{Multiple Constraints Lipschitz-only Safe BO (MCLoSBO)} (Algorithm~\ref{alg:cap}). This algorithm 
solves the problem above
based on assumptions that can be geometrically interpreted. Next, we introduce our assumptions, present our algorithm, and show how safety can be guaranteed.
\begin{figure}
    \centering
    \includegraphics{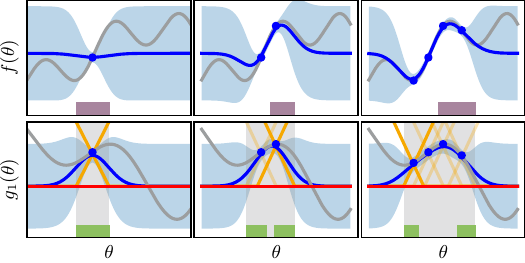}
    \vspace{-0.5cm}
    \caption{Illustration of MCLoSBO. The true performance function $f(\ctrl)$ (gray) and the true safety function $g_1(\ctrl)$ (gray) are queried in each iteration (blue dots). Based on the measurements, the safe set is determined following \eqref{safesetMCLoSBO} (orange). The functions are modeled by GPs (blue). The next query is chosen by evaluating the acquisition function in the maximizer set (violet) and expander sets (green).}
    \label{fig:MCLoSBO}
    \vspace{-0.5cm}
\end{figure}

\subsection{Assumptions}
The functional dependencies of $g_i$ on $\ctrl$ are unknown a priori. To provide safety guarantees, we need to make assumptions about the safety functions. To do this, we use the setting from \cite{fiedler2024safety}, which introduced Lipschitz-only Safe BO (LoSBO) for only one safety constraint (on the performance function). Here, we use the geometric assumption of Lipschitz continuity and extend it to the multiple constraints setting. 
Specifically, we assume:
\begin{assumption}
\label{ass:lipschitz}
All $g_{i}$ are $L_i$-Lipschitz continuous with Lipschitz constants $L_i$ so for all $\ctrl, \ctrl' \in \ctrldomain$, we have $|g_i(\ctrl)-g_i(\ctrl')| \leq L_i\|\ctrl-\ctrl'\|$.
\end{assumption}

This assumption can be geometrically interpreted as a Lipschitz cone at the safety function (Fig.~\ref{fig:MCLoSBO}, orange cones) and captures the maximal change of the function output for a given change in the inputs. In controller tuning applications, a reasonable estimate can be obtained, e.g., by a sensitivity analysis or based on simulation data.

In addition to a regularity assumption on $g_i$, we need an assumption on the noise in \eqref{eq:measurement_g}. In particular, we assume uniform bounded noise.
\begin{assumption}
\label{ass:noise}
For all $i \in \mathcal{I}_g$, there exists $E_i$ such that $|\epsilon_{i,n}| \leq E_i$ for all $n \geq 0$. 
\end{assumption}

A detailed discussion of these two assumptions, contrasting them with existing assumptions used in safe BO algorithms, can be found in \cite{fiedler2024safety}. 
Finally, we need an initial safe query $S_0 = \ctrl_0$ to start the optimization. This assumption is common in most safe BO algorithms \cite{sui15, Berkenkamp16, Sui2018, kirschner2019adaptive, Berkenkamp2020, fiedler2024safety}.

\subsection{Algorithm}
We propose the MCLoSBO algorithm, which extends LoSBO \cite{fiedler2024safety} to multiple safety constraints and leverages ideas from the SafeOpt-MC algorithm \cite{Berkenkamp2020}.
The two main challenges for a safe BO algorithm are to ensure safety and to trade off exploration and exploitation.
First, to guarantee safety in MCLoSBO, we define the safe set based on Assumptions \ref{ass:lipschitz} and \ref{ass:noise}
\begin{equation}
    \label{safesetMCLoSBO}
    S_n \coloneqq \bigcap\limits_{i\in \mathcal{I}_{g}} \bigcup\limits_{\ctrl\in S_{n-1}} \{\ctrl' \in \ctrldomain\:|\: y_{n,i}- E_i - L_i\|\ctrl-\ctrl'\| \geq 0\}.
\end{equation}
Under Assumptions \ref{ass:lipschitz} and \ref{ass:noise}, every $\ctrl_n \in S_n$ is safe according to \eqref{eq:safety}.

Next, we seek an evaluation criterion that expands the safe set while finding an optimum in the known safe set. 
Therefore, we introduce two subsets of the safe set, the set of potential maximizers $M_n \in S_n$ and the set of potential expanders $G_n \in S_n$. 
The definition of these sets corresponds to the definition in SafeOpt-MC. 
We use the tuning factor $\beta \in \mathbb{R}^+$ to define confidence intervals as
\begin{equation}
    \label{eq:confidence}
    Q_n(\ctrl,i) \coloneqq \left[ \mu_{n-1}(\ctrl,i) \pm \beta \sigma_{n-1}(\ctrl,i) \right].
\end{equation}
We define the lower uncertainty bound $l_{i,n}(\ctrl) \coloneqq \min{Q_n(\ctrl,i)}$ and upper uncertainty bound $u_{i,n}(\ctrl) \coloneqq \max{Q_n(\ctrl,i)}$.

In contrast to SafeOpt-MC, in MCLoSBO, the confidence intervals are not related to safety. 
The confidence bounds are only used for trading off the expansion and exploitation of the safe set.
Based on the confidence intervals, the set of potential expanders is defined as 
\begin{equation}
\label{eq:expandersmc}
    G_{n}(\ctrl) \coloneqq \{\ctrl \in S_{n}\:|\: e_{n}(\ctrl)  \geq 0\},
\end{equation}
with
\begin{equation}
\label{eq:expanders2mc}
    e_{n}(\ctrl) \coloneqq |\{\ctrl' \in \ctrldomain \setminus S_n \:|\: \exists i \in \mathcal{I}_g: u_{i,n}(\ctrl)-L_i\|\ctrl-\ctrl'\| \geq 0 \}|.
\end{equation}

The set of potential maximizers is
\begin{equation}
    M_{n} \coloneqq \{\ctrl \in S_{n}\:|\: u_{0,n}(\ctrl) \ge \max\limits_{\ctrl' \in S_{n}} l_{0,n}(\ctrl')\}.
\end{equation}

The acquisition function $\alpha:\ctrldomain \times \mathcal{I} \rightarrow \mathbb{R}$ is evaluated on the union of the maximizer and the expander sets. The function determines the next decision as the element with the highest uncertainty for all performance and safety functions. The acquisition function is defined by subtracting the lower and the upper uncertainty bound for each function. This yields
\begin{equation}
    \alpha(\ctrl,i)=u_{i,n}(\ctrl)-l_{i,n}(\ctrl).
\end{equation}

We can recover the setting of LoSBO by imposing $f=g_1$ and $q=1$. In this case, all assumptions on $g_1$ also hold for $f$. 
\begin{figure*}[t]
    \centering
    \includegraphics[width=0.9\textwidth]{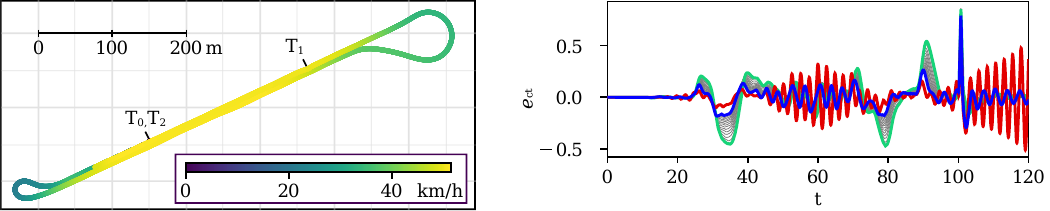}
    \caption{LEFT: Test track with the speed profile 
 of the experiment (color bar) in $\si{km/h}$. The simulation is started at the position $T_0$, and the disturbance is injected at $T_1=100$. The simulation ends at $T_2=120$ at the same position as $T_0$. This time scale corresponds to the time axis of the right plot. RIGHT: Plot of the cross-track error over time of an optimization with one parameter. The optimized controller is blue, the initial controller is green, and the intermediate optimization steps are grey. The red line represents parameters not queried in the optimization that violate $g_2$.}
    \label{fig:fasttrack}
\vspace{-0.5cm}
\end{figure*}

\subsection{Guarantees}
Combining the above assumptions, we can now provide safety guarantees for MCLoSBO, described in the following result. It is a direct adaption of Proposition
4.1 from \cite{fiedler2024safety} and its proof. 
\begin{proposition}
\label{prop:1}
Given the Assumptions \ref{ass:lipschitz} and \ref{ass:noise} for any choice of $\beta \in \mathbb{R}^+$, the MCLoSBO algorithm yields a sequence of safe inputs, i.e., $g_i(\ctrl_n) \geq 0$ for all $i=1, ..., q$ and $n \geq 1$.
\end{proposition}

\begin{proof}
It is enough to show that for all $ n \geq 0$ and $\ctrl \in S_n$, we have $g_i(\ctrl)\geq 0, \forall i \in \mathcal{I}_g = {1, ..., q}$. We use induction on $n$: For $n=0$, this follows by assumption. Now let $n \geq 1$ and assume $g_i(\ctrl)\geq 0, \forall i \in \mathcal{I}_g$ holds for all $\ctrl \in S_{n-1}$. Let $\ctrl$ be arbitrary in $S_n = \bigcap\limits_{i\in \mathcal{I}_{g}} \bigcup\limits_{\ctrl\in S_{n-1}} \{\ctrl' \in \ctrldomain\:|\: y_{i,n-1}- E_i -L_i\|\ctrl-\ctrl'\| \geq 0\}$. If $\ctrl \in S_{n-1}$, then $g_i(\ctrl)\geq 0$ follows from the induction hypothesis. Otherwise, for all $i \in \mathcal{I}_g$ we have 
\begin{multline}
    g_i(\ctrl)=g_i(\ctrl_n)+\epsilon_{i,n}-\epsilon_{i,n}+g_i(\ctrl)-g_i(\ctrl_n) \\
   \geq y_{i,n} - E_i - L_i\|\ctrl-\ctrl_n\| \geq 0, 
\end{multline}
where we used the $L_i$-Lipschitz continuity of $g_i$ and the noise bounds $|\epsilon_{i,n}| \leq E_i$ in the first inequality and the definition of $S_n$ in the second inequality. Altogether, we have that $\ctrl$ is in $S_{n}$. 
\end{proof}
\vspace{-0.3cm}
\begin{algorithm}
\caption{MCLoSBO}
\label{alg:cap}
\begin{algorithmic}[1]
\Require Domain $\ctrldomain$, $\mathcal{GP}_\ctrldomain(0,k)$, tuning factor $\beta \in \mathbb{R}^+$, Lipschitz constant $L_i$ for all $i \in \mathcal{I}_g$, noise bound $E_i$ for all $i \in \mathcal{I}_g$, initial safe set $S_0 \subseteq \ctrldomain$, initial data set $\mathcal{D}_0$
\For{$n=1,... $}
    \State \small $S_n \gets \bigcap\limits_{i\in \mathcal{I}_{g}} \bigcup\limits_{\ctrl \in S_{n-1}} \{\ctrl' \in \ctrldomain\:|\: y_{i, n-1}- E_i -L_i\|\ctrl-\ctrl'\| \geq 0\}$   
    \State \small $M_{n} \gets \{\ctrl \in S_{n}\:|\: u_{0,n}(\ctrl)  \geq \max\limits_{\ctrl' \in S_{n}} l_{0,n}(\ctrl')\}$
    \State \small $G_{n} \gets \{\ctrl \in S_{n}\:|\: e_{n}(\ctrl)  \geq 0\}$
    \State \small $\ctrl_{n} \gets \argmax \limits_{\ctrl \in G_{n} \cup M_{n}} \max\limits_{i \in \{0\} \cup \mathcal{I}_g} \alpha_n(\ctrl,i) $
    \State \small $y_{0,n} \gets f (\ctrl_n)+\epsilon_{0,n}$ 
    \State \small $y_{i,n} \gets g_i(\ctrl_n)+\epsilon_{i,n}$ for all $i \in \mathcal{I}_g$
    \State \small Augment data $\mathcal{D}_n = \{\mathcal{D}_{n-1}, \{(\ctrl_n, y_n)\}\}$
    \State \small Update GP with $\mathcal{D}_n$
\EndFor
\end{algorithmic}
\end{algorithm}
\vspace{-0.5cm}

\subsection{Practical adaptations}
Motivated by the controller tuning application, we introduce practical adaptations that can be included in the MCLoSBO algorithm. 
\subsubsection{Asynchronous optimization}
In many applications, there is no interval between consecutive experiments, and halting to await the optimizer's next query is inefficient. Hence, employing asynchronous observation with pending experiments, a variant of batch optimization, is a viable solution \cite{garnett}. 
We define a batch acquisition function $\alpha(\ctrl, \Tilde{\ctrl})$, where $\Tilde{\ctrl}$ is the fixed set of parameters that is currently being tested in the experiment but has no performance measure yet. 
When employing maximum variance as an acquisition function, we can incorporate a virtual data point $(\Tilde{\ctrl}, \mu_n(\Tilde{\ctrl}))$ into the GP before evaluating the acquisition function as the variance $\mathbf{\sigma_n}(\ctrl)$ is independent of $\mathbf{y}$. 
To ensure safety, we determine the safe set solely on the basis of the points with measured performance.  

\subsubsection{Hyperparameter optimization}
Optimization with safe BO algorithms requires selection of a GP kernel \cite{garnett}. In practice, selecting this kernel is not always straightforward. 
While some function properties, such as stationarity and ``roughness'', can be inferred from the context, specific values of the free hyperparameters, like lengthscales and outputscales, are difficult to determine. 
To address this, hyperparameter optimization can be applied to optimize the free hyperparameters based on measured data, for example, by using a marginal maximum likelihood approach \cite{rasmussen}.
The model selection effort increases substantially when modeling multiple performance and constraint functions.
Hyperparameter optimization for SafeOpt-MC as applied in \cite{doersche2021, König2021} led to safety violations, since safety in SafeOpt-MC directly depends on the covariance function and thus was highly influenced by the approximated model selection.
In contrast, safety in MCLoSBO is not dependent on the covariance function, as it only relies on the upper bound of the Lipschitz constant and the noise bound. 
Therefore, we can optimize the free hyperparameters during the tuning process, which reduces the effort required for the a priori model selection.

\section{Application Study}
\label{sec:application}
In this section, we evaluate the MCLoSBO algorithm for controller tuning in an automotive application. First, we assess the algorithm in simulation to evaluate the impact of the practical adaptations and to benchmark it against SafeOpt-MC.\footnote{The repository is available at \url{https://github.com/Data-Science-in-Mechanical-Engineering/mclosbo}} Second, we demonstrate MCLoSBO on a test vehicle.

\begin{figure*}[t]
    \centering
    \includegraphics[width=0.9\linewidth]{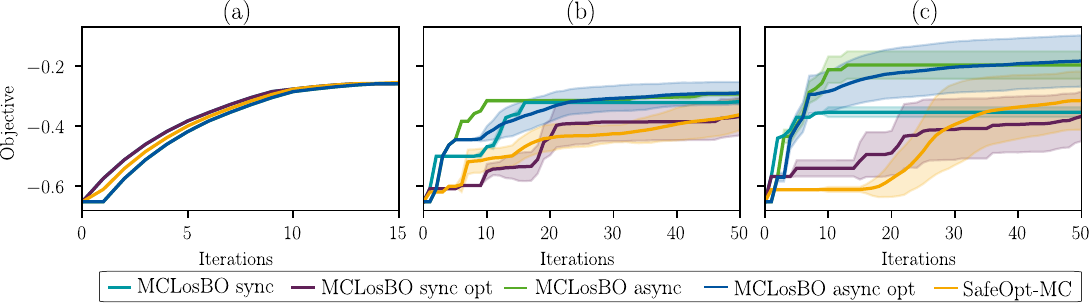}
    \caption{Simulation results for optimizing parameter sets of one (a), two (b), and three (c) parameters. In (a), the green and the blue line are overlapping.}
    \label{fig:results_sim}
    \vspace{-0.5cm}
\end{figure*}

\subsection{Experimental setup}
 We focus on a vehicle trajectory tracking control task, aiming to follow a pre-planned trajectory as closely as possible. This control task can be divided into two distinct tasks: longitudinal and lateral control. In this work, we use a given longitudinal controller and focus on optimizing the lateral controller. 
 For this, we define the cross-track error $e_{\mathrm{ct}}$ as the offset between the midpoint of the rear-axis and the reference point on the trajectory, and the course angle error $e_{\mathrm{ca}}$ as the difference between vehicle orientation and the trajectory orientation at the reference point. We use a modified version of the rear-axis tracking controller from \cite{Paden} with \expm{three} free controller parameters. As BO considers a black box optimization problem, we do not need information on the specific controller structure.

As presented in Fig. \ref{fig:fasttrack}, the test track is a closed course with two straights and two turns of different radii. The speed on the straights is about \expm{$55 \si{km/h}$} and it is reduced during cornering. When the vehicle is on the longest straight, a virtual disturbance on the cross-track error is injected at time $T_1$ to evaluate the controller stability. Accordingly, the measurement episode is separated into $[T_0,T_1)$, which includes the measurements on one straight and the turns; and $[T_1,T_2)$, which includes measurements for a fixed duration after the disturbance injection. 

To evaluate the controller's path-tracking performance, we formulate the objective function as a sum of the integrated tracking errors in lateral and angular directions normalized over time, plus the maximum cross-track error: 
\begin{equation}
    f(\ctrl)= \frac{\int_{T_0}^{T_1} |e_{\mathrm{ct}}(t)|+|e_{\mathrm{ca}}(t)|\, \mathrm{d}t}{T_1-T_0}  + \max_{t \in [T_0,T_1)} |e_{\mathrm{ct}}(t)|.
\end{equation}
We formulate \eqref{eq:optproblem} as a maximization problem, but our goal is to minimize the objective function, so we invert its sign for the optimization. 

To prevent the vehicle from leaving the track, our first safety constraint sets a threshold of \expm{$2 \si{m}$} for the maximum cross-track error for $g_1$. For the second constraint, we evaluate the robustness of the controller. Here, a virtual disturbance of \expm{$1 \si{m}$} cross-track error is injected at time $T_1$ when the vehicle drives straight. This emulates side wind or sudden steering motions both of which may cause unstable controller behavior. To measure this behavior, we examine the maximum yaw rate $\dot{\psi}$ after the disturbance injection. To reliably detect oscillation and aggressive controller behavior, we set the threshold for this criterion at $0.2\si{rad/s}$ as our second safety constraint.
Fig. \ref{fig:fasttrack} shows a safety violation of this constraint in red. The vehicle tracks the trajectory closely when cornering but oscillates when driving straight. This is a hazard when tuning parameters on the vehicle. In summary, the two safety constraints are:
\begin{align}
    g_1(\ctrl) &= -\max_{t \in [T_0,T_1)} |e_{\mathrm{ct}}(t)| + 2 \si{m} \\
    g_2(\ctrl) &= -\max_{t \in [T_1,T_2)} |\dot{\psi}(t)| + 0.2 \si{rad/s}.
\end{align}

The experimental setup imposes several additional restrictions. Firstly, stopping the vehicle after each round to adjust parameters is impractical; therefore, we apply the asynchronous optimization framework to prevent interruptions during execution. Secondly, the experiments themselves are expensive, so we assume a limited budget of iterations.

In our setting, the domain $\Theta$ is a rectangular box and the bounds based on expert knowledge and experience with the controller. We normalize the domain for each parameter to $[0,1]$ for all experiments. 

In all experiments, in simulation and on the vehicle, we model the unknown functions with Matern kernels \cite{rasmussen} with roughness parameter \expm{$\nu=5/2$}. Hyperparameters are lengthscale $l$, signal variance $\sigma_f$, and the noise variance $\sigma_d$. Table~\ref{tab:hyper} lists all algorithm parameters used for the experiments. 
We scale the objective function at each iteration using a Min-Max scalar with range $[0,1]$ on the observed values. For the hyperparameter optimization, we use a Gamma prior with $\mathbf{\Lambda} \sim \Gamma(a,b)$  with $\Gamma(3,10)$ on the length scales and $\Gamma(3,2)$ on the output scales. The Lipschitz constants are estimated by using a point grid in combination with domain knowledge.
\begin{table}[b]
    \vspace{-0.3cm}
    \centering
    \begin{tabular}{SSSSSS} \toprule
        {} & {$l$} & {$\sigma_f$} & {$\sigma_d$} & {$E_i$} & {$L_i$} \\ \midrule
        {$f$}  & 0.2 & 1 & 0.03 & 0.03 & /\\
        {$g_1$}  & 0.2  & 1 & 0.1 & 0.1 & 10\\
        {$g_2$}  & 0.2  & 0.2 & 0.01 & 0.1 & 3\\
        \bottomrule
    \end{tabular}
    \caption{Hyperparamter choices in the simulation study}
    \label{tab:hyper}

\end{table}
\subsection{Simulation results}
We perform simulation experiments with \expm{one, two, and three parameters}. We compare MCLoSBO with and without hyperparameter optimization in a standard setting (synchronous) and an asynchronous setting with pending experiments. As a baseline, we use the SafeOpt-MC algorithm from \cite{Berkenkamp2020}. For our experiments, key metrics include the maximum attained by the optimization and the number of safety violations.  For each experiment, we perform \expm{100} runs starting from the initial baseline controller. 

Prior work \cite{Wischnewski2019, doersche2021} has reported that safety violations can occur with SafeOpt-type algorithms, as is also shown and analyzed in detail in \cite{fiedler2024safety}. However, we observe no safety violations in this study, neither for SafeOpt-MC, nor for MCLosBO. In general, whether or not safety violations may occur in SafeOpt-MC depends on the kernel and hyperparameter choices, while MCLosBO provides safety guarantees based on geometric assumptions and irrespective of hyperparameter choices.

With regards to performance, the standard synchronous version of MCLosBO achieves similar results to SafeOpt-MC, while the asynchronous versions of MCLosBO perform better (see Fig.~\ref{fig:results_sim}).  In particular, when tuning one parameter (Fig.~\ref{fig:results_sim} (a)), all algorithms perform almost equally, as this optimization task is simple. For two and three parameters (Fig.~\ref{fig:results_sim} (b) and (c)), the adapted version with asynchronous optimization leads to significantly better results. 
Hyperparameter optimization results in a higher variance of the found optimum. In the synchronous setting, hyperparameter optimization results in a worse performance in this study. However, in the asynchronous setting, we observe slightly better performance in the setting with two parameters. 
Overall, this simulation study demonstrates that optimizing all three controller parameters can yield better optima.

\subsection{Application study: Tuning on the vehicle}
For the application of MCLosBO on the test vehicle (see Fig.~\ref{fig:summary}), we apply the same experimental setup used in simulation. All experiments are performed on an \expm{Opel Astra SportsTourer 1.4 ecoflex}. 
Our optimization takes a sequential approach. First, we optimize one parameter based on a conservative initial set. Second, we choose a reasonable safe set from the optimization with one parameter and optimize all three controller parameters. The Lipschitz constant of this experiment is determined by expert knowledge and the insights gained from simulation. In the first experiment, we use $L_1 = 4$ and $L_2 = 1.5$ and in the second $L_1 = 10$ and $L_2 = 1.5$. 

We note that during the experiments, measurement errors occurred due to flaws in the communication protocol. The falsely transmitted data points, which were recovered from logging data after the experiment, erroneously influenced the optimization. Nevertheless, this represents a realistic testing use case, and the method proved robust to these false measurements.

In the first experiment, we evaluate 15 iterations. The results are illustrated in Fig \ref{fig:plot_real_experiments}. We explore nearly the whole safe set and find an optimum close to the boundary of the safe set. No safety violations occurred, and the objective of the best controller parameters is $70\%$ better than the initial performance. In the second experiment, we evaluate 13 iterations for three parameters. No safety violations occurred, and the found optimum improved by about $28\%$ in comparison to the initial safe set of this optimization run. 

\begin{figure}
    \centering \includegraphics[trim = 0mm 4mm 0mm 0mm, clip, width=0.9\linewidth]{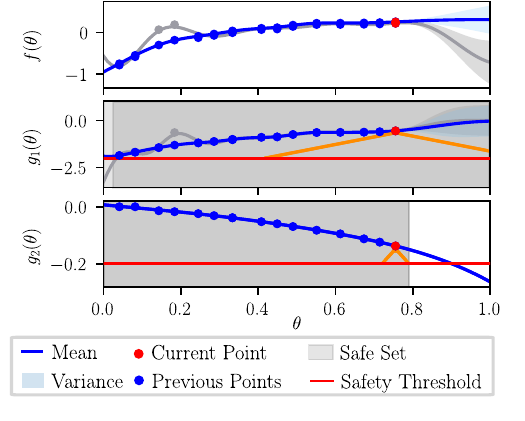}
    \caption{Results of the first vehicle experiment. No safety violations occur. The initial safe parameter is $\theta_0 = 0.3$. The red point is the last point of the optimization. These plots are generated in real time during the test and can be inspected by the test copilot. The gray points experiments represent the falsely measured points corrected based on the logging data.}
    \label{fig:plot_real_experiments}
    \vspace{-0.5cm}
\end{figure}

Overall, these experiments demonstrate that MCLoSBO can safely tune a controller directly on hardware. Despite being able to obtain solid insights from simulations, every application engineering task requires hardware tuning since there is always a gap between simulation and reality. 

\section{CONCLUSIONS}
In this work, we presented the MCLosBO algorithm, a safe BO algorithm with deterministic and interpretable safety guarantees. We applied the Lipschitz-safe algorithm to lateral trajectory-tracking control. While we selected this particular use case for demonstration purposes, the presented method can be used in a broad spectrum of learning problems where safety constraints arise. The simulation experiments demonstrated that MCLoSBO provides a reasonable safety mechanism and could outperform the baseline SafeOpt-MC. In a real vehicle tuning scenario, a controller was found that outperformed the given baseline while adhering to safety constraints, all without strong reliance on expert knowledge.

\addtolength{\textheight}{-15cm}  




\section*{ACKNOWLEDGMENT}
We thank P.~Brunzema, C.~Fiedler, E.~Sapozhnikova, F.~Solowjow, D.~Stenger, and A.~von~Rohr for very helpful discussions on this research and manuscript, and S.~Azirar for support with plots and implementing experiments.

 
\AtNextBibliography{\footnotesize}

\printbibliography

\end{document}